\documentclass[preliminary,copyright,creativecommons]{eptcs}
\providecommand{\event}{TARK 2025} 

\usepackage{iftex}

\ifpdf
  \usepackage{underscore}         
  \usepackage[T1]{fontenc}        
\else
  \usepackage{breakurl}           
\fi

\usepackage[utf8]{inputenc}
\usepackage{amsmath, amssymb, amsthm} 
\usepackage{tcolorbox}
\usepackage{tabularx, multirow} 
\usepackage{graphicx, tikz} 
\usetikzlibrary{fit,shapes,arrows,arrows.meta,calc,positioning,decorations}
\usepackage{verbatim} 
\usepackage{enumitem} 
\usepackage{cite}
\usepackage{color} 
\usepackage{booktabs}
\usepackage{threeparttable}
\usepackage{todonotes} 
\usepackage{amsmath,amssymb}
\usepackage{amsthm}
\usepackage[switch]{lineno}
\usepackage{pifont}
\usepackage{thmtools}
\usepackage{subcaption}
\usepackage{hyperref} 
\hypersetup{
    colorlinks=true,       
    linkcolor=green,        
    citecolor=purple,         
    urlcolor=black
}

\definecolor{blue}{RGB}{68,118,170} 
\definecolor{green}{RGB}{34,136,51} 
\definecolor{cyan}{RGB}{102,204,238}
\definecolor{coralred}{RGB}{238,102,119}
\definecolor{yellow}{RGB}{204,187,68}
\definecolor{purple}{RGB}{170,51,119}
\definecolor{grey}{RGB}{187,187,187}

\newcommand{\OMIT}[1]{}

\newcommand{\EP}[4]{
	\begin{center}
		\smallskip
		{\small 
			\begin{tabularx}{0.7\columnwidth}{l@{\hspace*{2mm}}l}
				\toprule
				\multicolumn{2}{c}{\sc{#2}}\label{#1} \\
				\midrule
				{\bf Given:}& \parbox[t]{0.5\columnwidth}{#3\vspace*{1mm}} \\
				{\bf Question:}& \parbox[t]{0.5\columnwidth}{#4\vspace*{.5mm}} \\ 
				\bottomrule
			\end{tabularx}
		}
		\smallskip
	\end{center}
}

\newtheorem{theorem}{Theorem}[section] 
\newtheorem{proposition}[theorem]{Proposition}

\newtheorem{claim}[theorem]{Claim}
\newtheorem{observation}[theorem]{Observation}

\theoremstyle{definition}

\newtheorem{example}{Example}

\newenvironment{claimproof}[1]{\par\noindent\textit{Proof:}\space#1}{\hfill $\lhd$\bigskip}

\newcommand{\p}{\ensuremath{\mathrm{P}}} 
\newcommand{\np}{\ensuremath{\mathrm{NP}}}
\newcommand{\conp}{\ensuremath{\mathrm{coNP}}}

\newcommand{\proofonlyif}{\smallskip\textit{Only if:\quad}}
\newcommand{\proofif}{\smallskip\textit{If:\quad}}

\newcommand{\myproblem}[1]{\textsc{#1}}

\def\AHG{AHG}
\newcommand{\minSF}{\textup{min-SF}}
\newcommand{\minEQ}{\textup{min-EQ}}
\newcommand{\minAL}{\textup{min-AL}}
\newcommand{\avgSF}{\textup{avg-SF}}
\newcommand{\avgEQ}{\textup{avg-EQ}}
\newcommand{\avgAL}{\textup{avg-AL}}

\newcommand{\val}{\mathrm{val}}
\newcommand{\util}{\mathrm{util}}

\def\EE{\mathcal{E}}

\newcommand{\friendaverage}[2]{\mathrm{avg}^F_{#1}(#2)}
\newcommand{\friendiaverage}[2]{\mathrm{avg}^{F+}_{#1}(#2)}

\newcommand{\friendmin}[2]{\mathrm{min}^F_{#1}(#2)}
\newcommand{\friendimin}[2]{\mathrm{min}^{F+}_{#1}(#2)}

\makeatletter
\providecommand*{\cupdot}{%
  \mathbin{%
    \mathpalette\@cupdot{}%
  }%
}
\newcommand*{\@cupdot}[2]{%
  \ooalign{%
    $\m@th#1\cup$\cr
    \sbox0{$#1\cup$}%
    \dimen@=\ht0 %
    \sbox0{$\m@th#1\cdot$}%
    \advance\dimen@ by -\ht0 %
    \dimen@=.5\dimen@
    \hidewidth\raise\dimen@\box0\hidewidth
  }%
}

\providecommand*{\bigcupdot}{%
  \mathop{%
    \vphantom{\bigcup}%
    \mathpalette\@bigcupdot{}%
  }%
}
\newcommand*{\@bigcupdot}[2]{%
  \ooalign{%
    $\m@th#1\bigcup$\cr
    \sbox0{$#1\bigcup$}%
    \dimen@=\ht0 %
    \advance\dimen@ by -\dp0 %
    \sbox0{\scalebox{2}{$\m@th#1\cdot$}}%
    \advance\dimen@ by -\ht0 %
    \dimen@=.5\dimen@
    \hidewidth\raise\dimen@\box0\hidewidth
  }%
}
\makeatother

\title{Solving Four Open Problems about Core Stability in Altruistic Hedonic Games}
\author{J\"{o}rg Rothe
\institute{Institut f\"ur Informatik, MNF\\
Heinrich-Heine-Universit\"at D\"usseldorf\\
D\"usseldorf, Germany}
\email{rothe@hhu.de}
\and
Ildik\'{o} Schlotter
\institute{Institute of Economics\\
  HUN-REN Centre for Economic and Regional Studies\\
  Budapest, Hungary}
\email{schlotter.ildiko@krtk.hun-ren.hu}
}
\def\titlerunning{Solving Four Open Problems about Core Stability in Altruistic Hedonic Games}
\def\authorrunning{J. Rothe \& I. Schlotter}
  
\begin{document}
\maketitle

\begin{abstract}
  Hedonic games---at the interface of cooperative game theory and computational social choice---are coalition formation games in which the players have preferences over the coalitions they can join.
  Kerkmann \emph{et al.}~\cite{ker-ngu-rey-rey-rot-sch-wie:j:altruistic-hedonic-games} introduced \emph{altruistic} hedonic games where the players' utilities depend not only on their own but also on their friends' valuations of coalitions.
  The complexity of the verification problem for core stability has remained open in four variants of altruistic hedonic games: namely, for the variants with average- and minimum-based ``equal-treatment'' and ``altruistic-treatment'' preferences.
  We solve these four open questions by proving the corresponding problems $\conp$-complete;
  our reductions rely on rather intricate gadgets in the related networks of friends.
\end{abstract}

\section{Introduction}

Hedonic games are at the interface of two areas: \emph{cooperative game theory}---in particular, they are special coalition formation games (CFGs)---and \emph{computational social choice}---as the players in a hedonic game have preferences over the coalitions (i.e., subsets of the players) they can join.
The ultimate goal of a hedonic game (and of CFGs in general) is to form a coalition structure (i.e., a partition of the player set into coalitions) that is stable in some sense.
To study which coalition structures are likely to form, various stability notions have been introduced among which \emph{core stability} perhaps is the most natural and most central one.
Suppose that there is a coalition structure~$\Gamma$ in a hedonic game but a group~$C$ of players are not satisfied with their coalitions in~$\Gamma$: Instead, they would prefer leaving their current coalition in~$\Gamma$ so as to form a new coalition together.
In such a case, we say that \emph{$C$ blocks~$\Gamma$}, and a coalition structure is said to be \emph{core-stable} if it is not blocked by any coalition.
For more background on hedonic games, on core stability as well as other stability concepts for them, and on the computational complexity of the verification and existence problems for hedonic games with respect to various stability notions, we refer to the book chapters by Aziz and Savani~\cite{azi-sav:b:handbook-comsoc-hedonic-games} and Bullinger \emph{et al.}~\cite{bul-elk-rot:b-2nd-edition:economics-and-computation-cooperative-game-theory} and to the survey by Woeginger~\cite{woe:c:core-stability-hedonic-games}.

Specifically, we will study \emph{altruistic} hedonic games, which were introduced by Kerkmann \emph{et al.}~\cite{ker-ngu-rey-rey-rot-sch-wie:j:altruistic-hedonic-games} (see also its predecessor by Nguyen \emph{et al.}~\cite{ngu-rey-rey-rot-sch:c:altruistic-hedonic-games}).
As noted by them (and previously already in the survey by Rothe~\cite{rot:c:thou-shalt-love-thy-neighbor-as-thyself-when-thou-playest-altruism-in-game-theory}), their work was inspired in part by the---apparently unrelated---work of biologists: Hare and Woods~\cite{har-woo:b:survival-of-the-friendliest} complement Darwin's celebrated thesis of \emph{``survival of the fittest''} with a novel insight into how evolution works, put forward as their thesis of \emph{``survival of the friendliest.''}
Indeed, they collect data and arguments (e.g., by comparing bonobos and chimpanzees, two species of great apes) showing that evolutionary success can also arise from a \emph{friendlier} behavior.
Transferring this idea to game theory, \emph{altruistic} behavior in games can surpass aggressive selfishness that only cares about maximizing one's own utilities, regardless of the impact on others, especially so when a group of players need to coalesce.
As surveyed by Rothe~\cite{rot:c:thou-shalt-love-thy-neighbor-as-thyself-when-thou-playest-altruism-in-game-theory}, a variety of altruistic models have been studied for noncooperative and, to a lesser extent, cooperative games, and Kerkmann \emph{et al.}~\cite{ker-ngu-rey-rey-rot-sch-wie:j:altruistic-hedonic-games,ngu-rey-rey-rot-sch:c:altruistic-hedonic-games} were the first to study altruism in hedonic games.
Later on, Schlueter and Goldsmith~\cite{sch-gol:c:super-altruistic-hedonic-games} extended their idea when they studied \emph{``super-altruistic hedonic games,''} and Kerkmann \emph{et al.}~\cite{ker-cra-rot:j:altruism-in-coalition-formation-games,ker-rot:c:altruism-in-coalition-formation-games} transferred their ideas and models to CFGs in general.

If game theory aims at modeling the behavior of players in the real world, thus providing a theoretical framework to realistically support their decision-making, it must take altruistic behavior into account.
Indeed, the models of Kerkmann \emph{et al.}~\cite{ker-ngu-rey-rey-rot-sch-wie:j:altruistic-hedonic-games,ngu-rey-rey-rot-sch:c:altruistic-hedonic-games,ker-cra-rot:j:altruism-in-coalition-formation-games,ker-rot:c:altruism-in-coalition-formation-games} are based on the \emph{friend-oriented preference extension} due to Dimitrov \emph{et al.}~\cite{dim-bor-hen-sun:j:core-stability-hedonic-games} where players divide the other players into friends and enemies, yielding a network of friends (a simple, undirected graph whose edges represent friendship relations among the players) that allows for a compact representation of hedonic games.\footnote{Representing hedonic games \emph{compactly} is important, as each of $n$ players can join $2^{n-1}$ coalitions, and all players need to express their preferences on all coalitions they can join. For a large variety of ways of succinctly representing hedonic games, we refer to, e.g., the book chapters by Aziz and Savani~\cite{azi-sav:b:handbook-comsoc-hedonic-games} and Bullinger \emph{et al.}~\cite{bul-elk-rot:b-2nd-edition:economics-and-computation-cooperative-game-theory} and also the work of Kerkmann \emph{et al.}~\cite{ker-lan-rey-rot-sch-sch:j:hedonic-games-with-ordinal-preferences-and-thresholds}.}

Unlike in the model of Dimitrov \emph{et al.}~\cite{dim-bor-hen-sun:j:core-stability-hedonic-games}, however, in the models of Kerkmann \emph{et al.}~\cite{ker-ngu-rey-rey-rot-sch-wie:j:altruistic-hedonic-games,ngu-rey-rey-rot-sch:c:altruistic-hedonic-games,ker-cra-rot:j:altruism-in-coalition-formation-games,ker-rot:c:altruism-in-coalition-formation-games} the players' utilities from a coalition are not solely determined by their own valuations; rather, also the (average or minimum) valuations of their friends are taken into account.
Specifically, three degrees of altruism are defined that differ depending on the order in which players refer to their own or their friends' valuations: In the \emph{selfish-first model}~(SF), the players first look at their own and then at their friends' valuations---the latter only if they themselves are indifferent; in the \emph{equal-treatment model}~(EQ), they treat their own and their friends' valuations equally at the same time; and in the \emph{altruistic-treatment} model~(AL), the players first consider their friends' valuations, and only in the case of indifference they decide according to their own valuation which coalition they prefer.

The above papers provide long lists of related work.
Instead of repeating them here, we just highlight some of the work that is most closely related to ours because it also studies core stability in hedonic games, such as the papers by Banerjee \emph{et al.}~\cite{ban-kon-soe:j:core-in-simple-coalition-formation-game},
Dimitrov \emph{et al.}~\cite{dim-bor-hen-sun:j:core-stability-hedonic-games},
Alcade and Romero-Medina~\cite{alc-rom:j:coalition-formation-and-stability},
Woeginger~\cite{woe:j:hardness-for-core-stability-in-hedonic-games},
Peters~\cite{pet:c:precise-hed-games}, 
Ohta \emph{et al.}~\cite{oht-bar-ism-sak-yok:c:core-stability-in-hedonic-games-among-friends-and-enemies}, and
Chen \emph{et al.}~\cite{che-csa-roy-sim:c:hedonic-games-with-friends-enemies-and-neutrals}.
Bullinger and Kober~\cite{bul-kob:c:loyalty-in-cardinal-hedonic-games}
introduced the notion of \emph{loyalty in cardinal hedonic games}, and their loyal variant of symmetric friend-oriented hedonic games is nothing other than the minimum-based altruistic hedonic games under EQ preferences.

For average- and minimum-based SF altruistic hedonic games, Kerkmann \emph{et al.}~\cite{ker-ngu-rey-rey-rot-sch-wie:j:altruistic-hedonic-games} have shown that it is $\conp$-complete to verify whether a given coalition structure is core-stable, and Kerkmann \emph{et al.}~\cite{ker-cra-rot:j:altruism-in-coalition-formation-games} showed $\conp$-completeness of the verification problem for average- and minimum-based SF altruistic CFGs, leaving these questions open for the other two degrees of altruism: average- and minimum-based EQ and AL preferences in both altruistic hedonic games and, more generally, altruistic CFGs.
Note that EQ and AL preferences are particularly interesting, as they are ``more altruistic'' than SF preferences.
Recently, Hoffjan \emph{et al.}~\cite{hof-ker-rot:c:core-stability-in-altruistic-coalition-formation-games} solved these open questions for altruistic CFGs, again showing $\conp$-completeness of the verification problem for average- and minimum-based EQ and AL preferences.\footnote{In fact, Kerkmann \emph{et al.}~\cite{ker-cra-rot:j:altruism-in-coalition-formation-games} and Hoffjan \emph{et al.}~\cite{hof-ker-rot:c:core-stability-in-altruistic-coalition-formation-games} consider \emph{sum-based} (not average-based) SF, EQ, and AL preferences.
  However, these are equivalent to (and a bit simpler than) average-based SF, EQ, and AL preferences for altruistic CFGs, as each player has the same number of friends, no matter which two coalition structures are being compared, so the denominator in the average can simply be omitted, leaving just the sum of their friends' valuations; cf.\ (\ref{eq:friendaverage}) in Section~\ref{sec:preliminaries}.}

We solve all four questions that remained open along this line of research: We show that for average- and minimum-based EQ and AL altruistic hedonic games, the verification problem also is $\conp$-complete.
Our proofs are based on constructing rather involved gadgets in the networks of friends of these games.

\section{Preliminaries}
\label{sec:preliminaries}

The goal of a \emph{coalition formation game} (CFG) is to partition a
finite
set $N$ of players into \emph{coalitions}, i.e., subsets of~$N$, yielding a \emph{coalition structure}, i.e., a partition of~$N$; the \emph{set of all possible coalition structures over~$N$} is denoted by~$\mathcal{C}_N$. 
Given a coalition structure $\Gamma \in \mathcal{C}_N$, the \emph{coalition containing player~$i$} is denoted by~$\Gamma(i)$.
We are interested in studying hedonic games, which constitute a special type of CFGs.
A \emph{hedonic game}~$(N, \succeq)$ is specified by the player set $N$ and a preference profile ${\succeq} = (\succeq_1,\ldots,\succeq_n)$, where
  $n=|N|$ and
${\succeq_i} \subseteq \mathcal{N}^i \times \mathcal{N}^i$ is player~$i$'s preference relation (i.e., a complete, weak order) over $\mathcal{N}^{i} = \{C\subseteq N \mid i\in C\}$, the set of all coalitions containing~$i$.

Specifically, for any two coalitions $C, D \in \mathcal{N}^{i}$, we write $C \succeq_i D$ to mean that \emph{player $i$ weakly prefers~$C$ to~$D$}; we write $C \succ_i D$ to mean that \emph{$i$ prefers $C$ to~$D$} (i.e., $C \succeq_i D$ and not $D \succeq_i C$); and we write $C \sim_i D$ to mean that \emph{$i$ is indifferent between $C$ and~$D$} (i.e., $C \succeq_i D$ and $D \succeq_i C$).
Since the players' preferences only depend on the coalitions containing them, $\succeq_i$ also induces a weak preference ranking of player $i \in N$ over the coalition structures in~$\mathcal{C}_N$: For $\Gamma, \Delta \in \mathcal{C}_N$, $\Gamma \succeq_i \Delta $ if and only if $\Gamma(i) \succeq_i \Delta(i)$.
For a hedonic game~$(N, \succeq)$, a partition $\Gamma \in \mathcal{C}_N$, and a coalition $C\subseteq N$, let $\Gamma_{C\rightarrow\emptyset}$ denote the coalition structure that results from $\Gamma$ when all players $i$ in $C$ leave $\Gamma(i)$ to form a new coalition, $C$, while all other players remain in their coalition in~$\Gamma$.\footnote{
  The notation $\Gamma_{X \to Y}$ usually denotes the coalition structure obtained from $\Gamma$ when a set~$X$ of players joins some coalition~$Y$; hence, our notation $\Gamma_{C \to \emptyset}$ reflects the intuition that establishing a new coalition $C$ can be thought of as all members of $C$ joining the empty coalition~$\emptyset$.
}
Formally,
\[
\Gamma_{C \to \emptyset}=\{C\} \cup \{C' \setminus C \mid C' \in \Gamma, C' \not\subseteq C\}.
\]

We say that
\emph{a coalition~$C$ blocks~$\Gamma$} if 
$
C
  \succ_i \Gamma(i)$
 for all players $i \in C$, and $\Gamma$ is \emph{core-stable} if no nonempty coalition blocks~$\Gamma$.

Kerkmann \emph{et al.}~\cite{ker-ngu-rey-rey-rot-sch-wie:j:altruistic-hedonic-games} (see also~\cite{ngu-rey-rey-rot-sch:c:altruistic-hedonic-games}) introduced \emph{altruistic hedonic games} (AHGs), which have later been generalized by Kerkmann \emph{et al.}~\cite{ker-cra-rot:j:altruism-in-coalition-formation-games} (see also~\cite{ker-rot:c:altruism-in-coalition-formation-games}) to \emph{altruistic coalition formation games} (ACFGs).
Both AHGs and ACFGs are based on the \emph{friend-oriented preference extension} due to Dimitrov \emph{et al.}~\cite{dim-bor-hen-sun:j:core-stability-hedonic-games}, which allows for a compact representation: Every player divides the other players into friends and enemies, which yields a
\emph{network of friends}---a simple, undirected graph $G = (N,\EE)$ whose vertices are the players and whose edges represent mutual friendship relations, whereas a missing edge indicates that these two players are enemies.
For $i\in N$, let $F_i = \{j\in N \mid \{i,j\} \in \EE\}$ be the set of \emph{$i$'s friends} and $E_i = N\setminus (F_i \cup \{i\})$ be the set of \emph{$i$'s enemies}, and define \emph{$i$'s friend-oriented valuation for a coalition~$C \subseteq \mathcal{N}^i$} by
\[
\val_i(C) = n\cdot |F_i \cap C| - |E_i \cap C|;
\]
recall that
$n = |N|$.
Hence, for any two coalitions $C, D \in \mathcal{N}^{i}$, \emph{player~$i$ friend-orientedly prefers $C$ to $D$}, denoted as $C \succ_i^F D$, exactly if $C$ contains more friends of $i$ than~$D$, or in case $C$ and $D$ contain the same number of players among $i$'s friends, if $C$ contains fewer enemies of $i$ than~$D$.
Note that $C \succ_i^F D$ if and only if $\val_i(C) > \val_i(D)$.

In the altruistic model of Kerkmann \emph{et al.}~\cite{ker-ngu-rey-rey-rot-sch-wie:j:altruistic-hedonic-games}, the players' utilities from a coalition are not only determined by their own valuations, but also the (average or minimum) valuations of their friends in the same coalition are taken into account.
Specifically, they define the following three degrees of altruism:
\begin{itemize}
  \item In the \emph{selfish-first model}~(SF), the players' preferences mainly depend on their own valuations, consulting their friends' valuations only if they are indifferent between two coalitions;
  \item in the \emph{equal-treatment model}~(EQ), the players weigh their own and their friends' valuations equally; and
  \item in the \emph{altruistic-treatment} model~(AL), the players first consider their friends' valuations, consulting their own valuations only in the case of indifference between two coalitions.
\end{itemize}

Formally, setting the minimum of the empty set to zero by convention, define the \emph{(friend-oriented) average} and \emph{minimum value of player~$i$'s friends in a coalition $C \in \mathcal{N}^{i}$ (without and with~$i$)} by
\begin{align}
  \label{eq:friendaverage}
\friendaverage{i}{C} &= \sum_{c \in C \cap F_i} \frac{\val_c(C)} {|C \cap F_i|}
&\textnormal{and }\qquad
\friendiaverage{i}{C} &=
\sum_{c \in (C \cap F_i) \cup \{i\}} \frac{\val_c(C)} {|(C \cap F_i)\cup\{i\}|},\\
  \label{eq:friendmin}
\friendmin{i}{C} &= \underset{ c \in C \cap F_i } \min \{ \val_c(C)\}
&\textnormal{and}\qquad
\friendimin{i}{C} &= \underset{c \in (C \cap F_i ) \cup \{i\}} \min \{ \val_c(C)\}.
\end{align}

Now, define player~$i$'s average-based and minimum-based \emph{utilities} from a coalition $C \in \mathcal{N}^{i}$ in an AHG according to the above three degrees of altruism:
\begin{linenomath}
\begin{align*}
  \util_{i}^\avgSF(C) &= w \cdot \val_i(C) + \friendaverage{i}{C}
  &\textnormal{and}\qquad
  \util_{i}^\minSF(C) &= w \cdot \val_i(C) + \friendmin{i}{C},\\
  \util_{i}^\avgEQ(C) &= \friendiaverage{i}{C}
  &\textnormal{and}\qquad
  \util_{i}^\minEQ(C) &= \friendimin{i}{C}, \textnormal{ and}\\
  \util_{i}^\avgAL(C) &=\val_i(C) + w \cdot \friendaverage{i}{C}
  &\textnormal{and}\qquad
  \util_{i}^\minAL(C) &=\val_i(C) + w \cdot \friendmin{i}{C},
\end{align*}
\end{linenomath}
where $w\geq n^4$ is a constant weight on some of the valuations that ensures that the SF utility is first determined by $i$'s valuation and the AL utility is first determined by $i$'s friends' valuations~\cite{ker-ngu-rey-rey-rot-sch-wie:j:altruistic-hedonic-games}.
Again, since players due to their hedonism care only about the coalitions containing them, their utilities from coalitions immediately induce their \emph{utilities from coalition structures}: $\util_{i}^\avgSF(\Gamma) = \util_{i}^\avgSF(\Gamma(i))$, etc.

Depending on which aggregation method and which degree of altruism is used, we denote our games as \emph{average-based} or \emph{min-based SF}, \emph{EQ}, or \emph{AL AHGs}.
We study the computational complexity of the \emph{verification problem} for them: 

\EP{VCSAHG}{Verification-of-Core-Stability-in-Altruistic-Hedonic-Games}
   {An AHG~$I$ and a coalition structure~$\Gamma$.}
   {Is $\Gamma$ core-stable for~$I$?}

\OMIT{
\begin{center}
\noindent
\fbox{
\begin{minipage}{0.9\textwidth}
\vspace{3pt}
    {\bf Verification of core stability in altruistic hedonic games:}
    \\[2pt]
    \begin{tabular}{l@{\hspace{5pt}}p{13cm}}
    {\bf Input:} &
    An AHG~$I$ and a coalition structure~$\Gamma$.
    \\
    {\bf Question:} & Is $\Gamma$ core-stable for~$I$?
    \end{tabular}
\end{minipage}
}
\end{center}
} 

\tikzset{
enemyarc/.style={ 
	color=coralred, 
	densely dashed,
	line width=1.8pt,
	}
}
\tikzset{
friendarc/.style={ 
	line width=1.8pt,
	}
}

\begin{figure*}[t]
\centering
\begin{subfigure}[c]{0.3\textwidth}
\centering
\begin{tikzpicture}
\node[inner sep=1.5pt] (a) at (0,0) {};
\node[inner sep=1.5pt] (b) at (1,0.7) {};
\node[inner sep=1.5pt] (c) at (1,-0.7) {};
\node[inner sep=1.5pt] (d) at (2,0) {};
\node[inner sep=1.5pt] (e) at (3,0) {};

\node[left, xshift=-2pt] at (a) {$a$};
\node[above, yshift=2pt] at (b) {$b$};
\node[below, yshift=-2pt] at (c) {$c$};
\node[below, yshift=-1pt] at (d) {$d$};
\node[below, yshift=-2pt] at (e) {$e$};

\draw[friendarc] (a)--(b);
\draw[friendarc] (a)--(c);
\draw[friendarc] (b)--(c);
\draw[friendarc] (b)--(d);
\draw[friendarc] (c)--(d);
\draw[friendarc] (d)--(e);

\fill (a) circle[radius=3pt,inner sep=1.5pt];
\fill (b) circle[radius=3pt,inner sep=1.5pt];
\fill (c) circle[radius=3pt,inner sep=1.5pt];
\fill (d) circle[radius=3pt,inner sep=1.5pt];
\fill (e) circle[radius=3pt,inner sep=1.5pt];
\end{tikzpicture}
\label{fig:example-graph}
\end{subfigure}
\hspace{10pt}
\begin{subfigure}[c]{0.62\textwidth}
\begin{tabular}{lccccc}
\toprule
& $a$ & $b,c$ & $d$ & $e$ \\
\midrule
\\[-10pt]
$\val_i(\Gamma)$ & 8 & 14 & 14 & 2 
\\[2pt]
\midrule
\\[-10pt]
$\util_{i}^\avgSF(\Gamma)$ & $8w+14$ & $14w+12$ & $14w+10$ & $2w+14$ 
\\[4pt]
$\util_{i}^\avgEQ(\Gamma)$ & $12$ & $12.5$ & $11$ & $8$ 
\\[4pt]
$\util_{i}^\avgAL(\Gamma)$ & $14w+8$ & $12w+14$ & $10w+14$ & $14w+2$ 
\\[2pt]
\midrule
\\[-10pt]
$\util_{i}^\minSF(\Gamma)$ & $8w+14$ & $14w+8$ & $14w+2$ & $2w+14$ 
\\[4pt]
$\util_{i}^\minEQ(\Gamma)$ & $8$ & $8$ & $2$ & $2$ 
\\[4pt]
$\util_{i}^\minAL(\Gamma)$ & $14w+8$ & $8w+14$ & $2w+14$ & $14w+2$
\\
\bottomrule
\end{tabular}
\end{subfigure}
\caption{Illustration for Example~\ref{ex:defs}.
The network of friends is depicted to the left.
The table on the right contains players' valuations and utilities in the coalition structure~$\Gamma=\{N\}$ under the six preference models considered.}
\label{fig:example}
\end{figure*}

\begin{example}
\label{ex:defs}
Consider a game~$I$ containing a set $N=\{a,b,c,d,e\}$ of players whose network of friends is depicted in Figure~\ref{fig:example}.
Note that $b$ and~$c$ have the same friends and enemies in~$N$, and thus have the same valuations and utilities as long as they are contained in the same coalition.

The coalition~$C=\{a,b,c,d\}$ yields the valuations and utilities for its players as shown in Table~\ref{tab:example}, where we fix some constant weight $w\geq n^4 = 5^4 = 625$ ensuring the right order for SF and AL utilities.
Consider the coalition structure~$\Gamma = \{N\}$ containing only the grand coalition~$N$, which consists of all players.
Observe that $C$ blocks $\Gamma$ in both the minimum-based EQ and AL models, but does not block~$\Gamma$ under the remaining four models.

\def\yes{\textcolor{green}{\ding{52}}}
\def\no{\textcolor{purple}{\ding{55}}}
\begin{table}[h]
\begin{center}
\begin{tabular}{lccccc}
\toprule
& $a$ & $b,c$ & $d$ & Does $C$ block~$\{N\}$? \\
\midrule
\\[-10pt]
$\val_i(C)$ & 9 & 15 & 9 
\\[2pt]
\midrule
\\[-10pt]
$\util_{i}^\avgSF(C)$ & $9w+15$ \yes & $15w+11$ \yes & $9w+15$ \no  & no 
\\[4pt]
$\util_{i}^\avgEQ(C)$ &  $13$ \yes & $12$ \no & $13$ \yes & no
\\[4pt]
$\util_{i}^\avgAL(C)$ & $15w+9$ \yes & $11w+15$ \no & $15w+9$ \yes & no
\\[2pt]
\midrule
\\[-10pt]
$\util_{i}^\minSF(C)$ & $9w+15$ \yes & $15w+9$ \yes & $9w+15$ \no & no 
\\[4pt]
$\util_{i}^\minEQ(C)$ & $9$ \yes & $9$ \yes & $9$ \yes & yes
\\[4pt]
$\util_{i}^\minAL(C)$ & $15w+9$ \yes & $9w+15$ \yes & $15w+9$ \yes & yes
\\
\bottomrule
\end{tabular}
\end{center}
\caption{The valuations and utilities of players in the coalition~$C=\{a,b,c,d\}$. The sign \yes\ means that the given player~$i \in N$ prefers~$C$ to~$N$ under the given preference model, i.e., $C \succ_i N$, while
the sign \no\ means that $N \succ_i C$.}
\label{tab:example}
\end{table}

This shows that $\Gamma$ is \emph{not} in the core of~$I$ if $I$ is interpreted as a min-based EQ or AL AHG. 
By contrast, it is not hard to verify that $\Gamma$ is in the core of~$I$ if $I$ is interpreted as an average-based SF, EQ, or AL AHG or a min-based SF AHG.
\end{example}

Since our AHGs are compactly represented by their  (undirected) friendship graphs~${G=(N,\EE)}$, we also need some graph-theoretic notation.
For any subset $M \subseteq N$ of the vertices of~$G$, let $G[M]$ denote the \emph{subgraph of $G$ induced by~$M$}.
For each vertex $i \in N$, let $\delta(i) \subseteq \EE$ denote the set of edges incident to~$i$ in~$G$.
Note that $\delta(i) = \{\{i,j\} \mid j \in F_i\}$ gives player~$i$'s friendship relations in the game.

Referring to some standard textbooks~\cite{gar-joh:b:int,pap:b:complexity,rot:b:cryptocomplexity}, we assume the reader to be familiar with the basic notions of computational complexity theory, such as the complexity classes $\p$ (\emph{deterministic polynomial time}), $\np$ (\emph{nondeterministic polynomial time}), and $\conp = \{\overline{L} \mid L \in \np\}$ (the \emph{class of complements of $\np$ problems}).
Our reductions used to show $\conp$-hardness are based on the standard \emph{polynomial-time many-to-one reducibility}, and problems that are $\conp$-hard and in $\conp$ are said to be \emph{$\conp$-complete}.

\section{Core Stability in AHGs}

Each of our reductions is from the $\np$-complete  problem \myproblem{Clique} where we are given a graph~$H=(V,E)$ with an integer~$k$, and the question is whether $H$ contains a clique of size~$k$, i.e., a complete graph on~$k$ vertices, as a subgraph.
Before delving into the details of our proofs, let us sketch a high-level description of the approach shared by all of them. 

Given an instance $(H,k)$ of \textsc{Clique}   with input graph~$H=(V,E)$, each of our reductions defines an AHG~$I$ and a coalition structure~$\Gamma$ for~$I$ so that $H$ has a clique of size~$k$ if and only if $\Gamma$ is \emph{not} core-stable for~$I$. 
These reductions have a common approach: They construct
\emph{vertex}, \emph{edge}, and \emph{incidence gadgets} reflecting the structure of~$H$. Namely, for each vertex~$v \in V$, for each edge~$e \in E$, and for each pair~$(v,e)$ where $v$ is a vertex incident to some edge~$e$ in~$H$, 
we define a certain \emph{gadget}.
Although vertex, edge, and incidence gadgets may differ from each other, each of these gadgets contains a distinguished player that represents the given vertex~$v \in V$, edge~$e \in E$, or pair $(v,e) \in \{(v,e \mid v \in V, e \in \delta(v)\}$; this player is called the \emph{vertex}, \emph{edge}, or \emph{incidence player} corresponding to~$v$, to~$e$, or to $(v,e)$, respectively. 

The importance of distinguished players comes from the fact that, in each gadget, only the distinguished player has friends outside the gadget. 
As each of the gadgets forms a coalition in the coalition structure~$\Gamma$ for the constructed AHG, the properties of the gadgets will ensure that from among all players in the constructed gadgets, only distinguished players might be willing to form a coalition~$C$ that blocks~$\Gamma$. For this to happen, the distinguished players need not only to have enough friends within~$C$ but, due to the altruistic preference models we use, their friends in~$C$ also need a certain number of friends within~$C$. 
Moreover, some of the distinguished players will only be able to achieve a better utility value in~$C$ than in~$\Gamma$ if they have fewer enemies in~$C$ than they have within their gadget; hence, such players will impose an upper bound on the size of~$C$.

Therefore, the computational intractability of finding a blocking coalition is caused by the need to ensure a lower bound on the number of friends for certain players in the blocking coalition~$C$, while also respecting an upper bound on the size of~$C$.
Translating this into graph-theoretic terminology, the task is to find a subgraph in the network of friends that satisfies certain lower bounds on vertex degrees while respecting an upper bound on the number of vertices in the subgraph.
The following observation offers a flexible way to deduce the presence of a clique from such requirements.

\begin{observation}
\label{obs:clique}
If a simple graph~$H=(V,E)$ satisfies $|V|+\alpha|E|\leq k+\alpha\binom{k}{2}$ for some $\alpha\geq 1$ and each of its vertices has degree at least~$k-1$, then $H$ is a clique of size~$k$.
\end{observation}

\begin{proof}
  Let $v$ be any vertex in the graph.
  Since $v$ has degree at least~$k-1$, the graph contains at least~$k$ vertices.
  Since each of them has degree at least~$k-1$, the number of edges is at least $\frac{k(k-1)}{2}$, so we get $|V|+\alpha|E| \geq k+\alpha\binom{k}{2}$.
  Thus all of these inequalities must hold with equality, so the number of edges must be \emph{exactly}~$\binom{k}{2}$, and the number of vertices must
  be \emph{exactly}~$k$. This means that the graph is indeed a clique of size~$k$.
\end{proof}

\subsection{Min-Based EQ and AL AHGs}

We show that verifying core stability in min-based AHGs is $\conp$-complete, using the same construction for both min-based EQ and min-based AL preferences.

\begin{theorem}
\label{thm:minEQAL-corestable-coNPc}
Verifying core stability in a min-based EQ or AL  \AHG\ is $\conp$-complete.
\end{theorem}

\begin{proof}
  Both problems clearly are in $\conp$ because one can verify in polynomial time that a given coalition is blocking.
  To prove their $\conp$-hardness, we present a reduction from \myproblem{Clique} to their complements.
  Let our input for \myproblem{Clique} be the graph $H=(V,E)$ and integer~$k$.
  If $k$ is even, then we add an additional vertex connected to every vertex of~$H$; the obtained graph has a clique of size~$k+1$ if and only if $H$ contains a clique of size~$k$.
  Therefore, without loss of generality, we may assume that $k$ is odd.

\smallskip
\noindent
{\bf Construction.}
We construct from $(H,k)$ a min-based EQ or AL AHG over player set~$N$ and an underlying friendship graph~$G$.
Let us start by defining a \emph{$(k-1,k')$-circulant gadget} for $k'=k\binom{k}{2}+k+1$ as follows:
It contains~$k'$ players arranged along a cycle~$Q$ of length~$k'$, and each player~$i$ on~$Q$ is friends with those $k-1$ players on~$Q$ who are at a distance of at most $\frac{k-1}{2}$ away from~$i$ along $Q$ (not including~$i$ itself); see Figure~\ref{fig:circulant} for an illustration.
We refer to~$Q$ as the \emph{base cycle} of the gadget. 

We now introduce vertex, edge, and incidence gadgets; each of these gadgets will be a $(k-1,k')$-circulant gadget. 
For each vertex $v \in V$, we add a \emph{vertex gadget} over player set~$P_v$, with a special \emph{vertex player}~$v' \in P_v$ corresponding to the vertex~$v$. 
Similarly, for each edge $e \in E$, we introduce an \emph{edge gadget} over player set~$P_e$,
containing the \emph{edge player}~$e' \in P_e$ corresponding to the edge~$e$. 
Next, for each vertex~$v \in V$ and for each edge
$e \in \delta(v)$, we introduce an \emph{incidence gadget} over player set~$P_{v,e}$, containing the \emph{incidence player}~$b_{v,e}$. 
Additionally, we define a set~$A_e$ of $k-3$ \emph{dummy players} for each edge~$e \in E$.
We will use the notation $V'=\{v' \mid v \in V\}$, $E'=\{e' \mid e \in E\}$, 
$B=\{b_{v,e} \mid v \in V,\ e \in \delta(v)\}$, and
$A=\bigcup_{e \in E} A_e$.
The total set of players is
\[
N=\bigcupdot_{v \in V} P_v \cup \bigcupdot_{e \in E} \left(P_e \cup A_e\right) \cup \bigcupdot_{v \in V,\ e \in \delta(v)} P_{v,e}.
\]

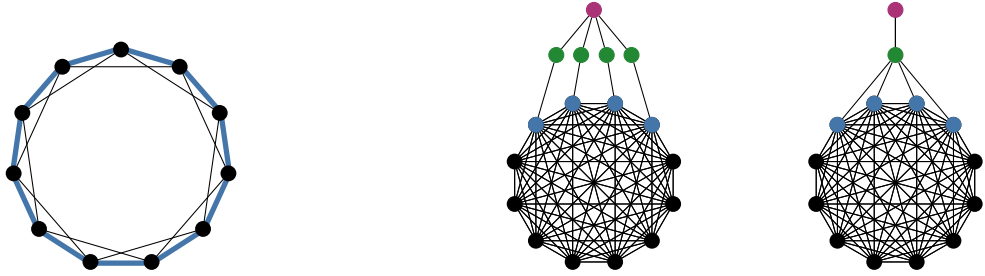
\begin{figure*}[t]
\centering
\begin{subfigure}[t]{0.4\textwidth}
\centering
\begin{tikzpicture}
\node[draw=none,line width=2pt,minimum size=80pt,regular polygon,regular polygon sides=11] (a) {};
\node[draw,color=blue,line width=2pt,minimum size=83pt,regular polygon,regular polygon sides=11] (b) {};
\foreach \x in {1,2,...,11}
  \fill (a.corner \x) circle[radius=3pt,inner sep=1.5pt];
\foreach \x in {3,4,...,11}
  \draw (a.corner \x)--(a.corner \the\numexpr\x-2);
\draw (a.corner 10)--(a.corner 1); 
\draw (a.corner 11)--(a.corner 2);
\end{tikzpicture}
\caption{The friendship graph of a $(k-1,k')$-circulant gadget for $k=5$ with $k'=11$ players, with the base cycle shown with bold, blue lines.}
\label{fig:circulant}
\end{subfigure}
\hspace{0.5cm}
\begin{subfigure}[t]{0.55\textwidth}
\centering
\begin{tikzpicture}
\node[draw=none,color=blue,line width=2pt,minimum size=60pt,regular polygon,regular polygon sides=12] (a) {};
\foreach \x in {1,2,...,12}
  \foreach \y in {1,2,...,12}
	  \draw (a.corner \x)--(a.corner \y);
\foreach \x in {1,2,...,12}
  \fill (a.corner \x) circle[radius=3pt,inner sep=2pt];

\node[circle,draw,fill,color=green,radius=3pt,inner sep=2pt] at (0.17,1.7) (m3) {};
\node[circle,draw,fill,color=green,radius=3pt,inner sep=2pt] at (-0.17,1.7) (m2) {};
\node[circle,draw,fill,color=green,radius=3pt,inner sep=2pt] at (0.5,1.7) (m4) {};
\node[circle,draw,fill,color=green,radius=3pt,inner sep=2pt] at (-0.5,1.7) (m1) {};

\node[circle,draw,fill,color=purple,radius=3pt,inner sep=2pt] at (0,2.3) (t) {};

\draw (a.corner 3)--(m1)--(t);
\draw (a.corner 2)--(m2)--(t);
\draw (a.corner 1)--(m3)--(t);
\draw (a.corner 12)--(m4)--(t);

\foreach \x in {1,2,3,12}
  \fill[color=blue] (a.corner \x) circle[radius=3pt];
\end{tikzpicture}
\hspace{1.5cm}
\begin{tikzpicture}
\node[draw=none,color=blue,line width=2pt,minimum size=60pt,regular polygon,regular polygon sides=12] (a) {};
\foreach \x in {1,2,...,12}
  \foreach \y in {1,2,...,12}
	  \draw (a.corner \x)--(a.corner \y);
\foreach \x in {1,2,...,12}
  \fill (a.corner \x) circle[radius=3pt,inner sep=2pt];

s\node[circle,draw,fill,color=green,radius=3pt,inner sep=2pt] at (0,1.7) (m0) {};

\node[circle,draw,fill,color=purple,radius=3pt,inner sep=2pt] at (0,2.3) (t) {};

\draw (a.corner 3)--(m0)--(t);
\draw (a.corner 2)--(m0)--(t);
\draw (a.corner 1)--(m0)--(t);
\draw (a.corner 12)--(m0)--(t);

\foreach \x in {1,2,3,12}
  \fill[color=blue] (a.corner \x) circle[radius=3pt];
\end{tikzpicture}
\caption{The friendship graphs of a normal (i.e., non-pinched) and a pinched $(d,k')$-dome gadget for $d=4$ with $k'=17$ players. The top, mid, and fringe players are depicted as purple, green, and blue circles, respectively.}
\label{fig:dome}
\end{subfigure}
\caption{Illustration of our gadgets in the proofs of Theorems~\ref{thm:minEQAL-corestable-coNPc}, \ref{thm:avgEQ-corestable-coNPc}, and~\ref{thm:avgAL-corestable-coNPc}.}
\end{figure*}

Besides the friendships within gadgets, 
we let each incidence
player~$b_{v,e} \in B$ be friends with the vertex player~$v'$ and the edge player~$e'$, so $G[V' \cup E' \cup B]$ 
can be obtained from~$H$ by subdividing each of its edges once.
Next, for each edge $e=xy \in E$, we let the dummy players in~$A_e$ be friends with players~$e', b_{x,e}$, and~$b_{y,e}$.
Finally, for each edge~$e \in E$, we let $G[A_e]$ be a clique of size~$k-3$,
so all players in~$A_e$ are friends with each other.
This completes the definition of our \AHG.

Let us define the coalition structure~$\Gamma$ as 
\[
\Gamma = \{P_v \mid v \in V\} \cup \{P_e \mid e \in E\}  \cup \{P_{v,e} \mid v \in V,e \in \delta(v)\} \cup \{A_e \mid e \in E\},
\]
so $\Gamma$ contains the player set of each vertex, edge, and incidence gadget as well as the set of dummy players associated with any edge as a coalition.
Note that, for $n=|N|$,
each player~$i$'s valuation of~$\Gamma$ is 
\[
\val_i(\Gamma) = \left\{
\begin{array}{ll}
(k-1) \cdot n-(k'-k) & \textrm{ if } i \in N \setminus A,
\\
(k-4) \cdot n & \textrm{ if } i \in A.
\end{array}
\right.
\]
Therefore, both in the min-based EQ and AL models, the utility of each player equals its valuation.

\smallskip
\noindent
{\bf Proof of correctness.}
We claim that $\Gamma$ is \emph{not} core-stable if and only if $H$ contains a clique of size~$k$.

\proofonlyif
Let us first assume that $\Gamma$ is \emph{not} core-stable; we show that $H$ contains a clique of size~$k$. 
The following claim captures the key property of our gadgets. 

\begin{claim}
\label{clm:circulant}
Let $P$ be
a set of players of a $(k-1,k')$-circulant gadget in the constructed instance whose unique player having friends outside~$P$ is $p^\star$ (a vertex, edge, or incidence player).
If some coalition $C \subseteq N$ blocks~$\Gamma$, then $C \cap P \subseteq \{p^\star\}$. 
\end{claim}

\begin{claimproof}
Let us show first that each player~$i \in C \cap (P \setminus \{p^\star\})$ must have at least~$k-1$ friends in~$C$. 
Suppose for the sake of contradiction that $i$ has at most~$k-2$ friends in~$C$.
On the one hand, if $i$ has a friend~$j \in C$, then $j \in C\cap P$ and,
for $\varphi \in \{\minEQ,\minAL\}$,
the utility of~$j$ in~$\Gamma_{C \to \emptyset}$ is at most the valuation of~$i$:
\[
\util_j^\varphi(\Gamma_{C \to \emptyset}) \leq
\val_i(\Gamma_{C \to \emptyset}) \leq
(k-2) \cdot n <
\util_j^\varphi(\Gamma),
\] 
which contradicts $j \in C$.
On the other hand, if $i$ has no friends in~$C$, then $\util_i^\varphi(\Gamma_{C \to \emptyset})=0<\util_i^{\varphi}(\Gamma)$, which contradicts $i \in C$. 

Now, to see the statement of the claim, assume for the sake of contradiction that $C$ contains some player~$i \in P \setminus \{p^\star\}$.
Let $Q$ denote the base cycle of the $(k-1,k')$-circulant gadget with players~$P$.
If $P \not\subseteq C$
were true, then consider the path in~$Q[C]$ that contains~$i$ (a subpath of the base cycle);  
let player~$j$ be one endpoint of this path that does not coincide with~$p^{\star}$ (such a player~$j$ exists because $i \neq p^\star$).
Then~$j$ has
fewer than~$k-1$ friends in~$C$ because one of
$j$'s friends, though adjacent to~$j$ on~$Q$, is not in~$C$.
This is not possible by the argument given in the previous paragraph, so $P \subseteq C$ follows.

To see that $P \subseteq C$ is not possible either, consider some player~$i$ in~$P \setminus\{p^\star\}$ that is not a friend of~$p^\star$; by $k'>k$, such a player exists.
Assuming that $P \subseteq C$,
we get that $i$ and all friends of~$i$ have valuation exactly~$(k-1)n - (|C|-k)$ in~$\Gamma_{C \to \emptyset}$, which is at most $(k-1)n-(k'-k)=\util^{\varphi}_i(\Gamma)$  by $|C| \geq k'$.
This contradicts $i \in P \subseteq C$.
\end{claimproof}

We will also need the following simple claim.

\begin{claim} 
\label{clm:min-degree}
Let $C \subseteq N$ be a coalition blocking~$\Gamma$.
Each player~$i \in C$ has at least $k-1$ friends in~$C$.
\end{claim}

\begin{claimproof} 
  First, if some player~$i$ has a non-dummy friend~$j$
  in~$C$ (i.e., a friend $j \in C \setminus A$), then $i$ must have at least~$k-1$ friends in~$C$, as otherwise,
for $\varphi \in \{\minEQ,\minAL\}$, we have
$\util_j^{\varphi}(\Gamma_{C \to \emptyset}) \leq \val_i(\Gamma_{C \to \emptyset}) \leq (k-2)n<\util_j^{\varphi}(\Gamma)$ which contradicts~$j \in C$.
  In particular, since dummy players in~$C$ must have non-dummy friends in~$C $ (as $C \subseteq A$ is not possible), and each dummy player has $k-1$ friends in total, we get that whenever $C \cap A_e \neq \emptyset$ for some $e =xy\in E$, then $C$ must contain all players in~$\{e',b_{x,e},b_{y,e} \} \cup A_e$. 

  To prove the claim, it remains to show that every non-dummy player~$i$ in~$C$ has a non-dummy friend in~$C$. First, every player in~$C$ must have some friend in~$C$, as otherwise
this player's utility is zero in~$\Gamma_{C \to \emptyset}$.
  For the sake of contradiction, assume that some non-dummy player~$i \in C$  has only dummy friends in~$C$.
  Then, by construction, $i$ must be among the players in $\{e',b_{x,e},b_{y,e} \}$ for some edge $e=xy \in E$ for which $C \cap A_e \neq \emptyset$. 
  Since $e'$ is friends with both~$b_{x,e}$ and~$b_{y,e}$, it follows that $i$ has a friend in~$\{e',b_{x,e},b_{y,e}\}$.
  However, by the argument given in the previous paragraph, $C \cap A_e \neq \emptyset$ implies $\{e',b_{x,e},b_{y,e} \} \subseteq C$, showing that $i$ has a non-dummy friend in~$C$.
  This contradiction proves the claim.
\end{claimproof}

Let $C \subseteq N$ be a coalition that blocks~$\Gamma$.
Using Claim~\ref{clm:circulant}, we obtain that  $C \subseteq V' \cup E' \cup B \cup A$.
Let us create the subgraph~$H_C$ of~$H$ that contains a given vertex or edge if and only if the corresponding vertex or edge player is contained in~$C$.
We
now show that $H_C$ is a well-defined subgraph of~$H$.
First, recall that each incidence player~$b_{v,e}$  has exactly $k-1$ friends in~$ V' \cup E' \cup B \cup A$, namely the players~$v'$, $e'$, and the~$k-3$ players in~$A_e$; therefore, 
if $b_{v,e}$ is contained in~$C$, then 
due to Claim~\ref{clm:min-degree} all friends of~$b_{v,e}$ must be in~$C$ as well. 
Keeping this in mind, assume that $e' \in C$ for some edge~$e=xy \in E$. Since $e'$ also has exactly $k-1$ friends in~$ V' \cup E' \cup B \cup A$, namely those in~$A_e$ together with the two incidence player~$b_{x,e}$ and~$b_{y,e}$,
all of them must be in~$C$ by Claim~\ref{clm:min-degree}.
However, $\{b_{x,e},b_{y,e}\} \subseteq C$ in turn implies $\{x',y'\} \subseteq C$, so $H_C$ is indeed well-defined.

Using again Claim~\ref{clm:min-degree}, we know that each vertex in~$H_C$ has degree at least~$k-1$.
To see this, note that $x' \in C \cap V'$ implies that there are at least $k-1$
incidence-player friends of~$x'$ in~$C$, together with their edge-player friends, which means that there are at least~$k-1$ edges incident to~$x$ in~$H_C$. 
To prove that $H_C$ is a clique of size~$k$, we will use  Observation~\ref{obs:clique}. 

By our previous arguments, $e' \in C \cap E'$ implies that all 
 $k-1$ friends of~$e'$ in~$B \cup A$ are also in~$C$.
Since no player in~$B \cup A$ has two edge players as a friend, we get 
\[
|C| = |C \cap V'|+|C \cap E'|+|C \cap (B\cup A)| \geq 
|C \cap V'|+k|C \cap E'|.
\]
Consider now an edge player~$e'$ in~$C$.
Since every friend of~$e'$ in~$C$, as well as~$e'$ itself, has exactly $k-1$ friends in~$C$, the utility of~$e'$ is $\util_{e'}^{\varphi}(\Gamma_{C \to \emptyset})=(k-1)n-(|C|-k)$ for $\varphi \in \{\minEQ,\minAL\}$, which exceeds~$\util_{e'}^{\varphi}(\Gamma)$ only if $|C|<|k'|$.
Thus $|C| \leq k \binom{k}{2}+k$, which implies $|C \cap V'|+k|C \cap E'|\leq k+k\binom{k}{2}$.
Applying Observation~\ref{obs:clique} with $\alpha=k$, it follows that $H_C$ is a clique of size~$k$.

\proofif
Let us now show that if $H$ contains a clique on a set~$K$ of $k$ vertices, then there exists a coalition $C \subseteq N$ that blocks~$\Gamma$.
Let $C$ contain the vertex and edge players corresponding to the vertices and edges of the clique~$K$, together with all friends of these edge players in~$A \cup B$, so $|C|=k+k\binom{k}{2}$.
The utility of each player in~$\Gamma_{C \to \emptyset}$ is then
$(k-1)n-(|C|-k)$, which is more than its utility in~$\Gamma$, because $|C|=k+k\binom{k}{2}<k'$. 
Therefore, $C$ blocks~$\Gamma$, and so $\Gamma$ is \emph{not} core-stable.
\end{proof}

\subsection{Average-Based EQ and AL AHGs}

To show the $\conp$-completeness of verifying core stability in an average-based EQ or AL \AHG, 
we introduce a new type of gadgets. 
For some integers $d \geq 1$ and $k'>2d+1$, we define a \emph{$(d,k')$-dome gadget} as follows: It contains $k'$ players, among them a \emph{top player}~$p^\star$ with friends $p_1,\dots,p_d$ that we call \emph{mid players}, and with the remaining $k'-d-1$ players forming a clique in the friendship graph that we call the \emph{base clique}.
Additionally, for each mid player~$p_i$, we select a player~$p'_i$ in the base clique in a way that the players $p'_1,\dots,p'_d$, called \emph{fringe players}, are all distinct, and we let $p_i$ be friends with~$p'_i$.

We further introduce a modification of this gadget: Let a \emph{pinched $(d,k')$-dome} be obtained from a $(d,k')$-dome gadget by the identification of all $d$ mid players into a single mid player.

See Figure~\ref{fig:dome} for an illustration of such gadgets. 
The following claim captures their key property. 

\begin{proposition}
\label{prop:dome}
Suppose that an instance of an average-based EQ or AL \AHG\ over~$n$ players
contains a (possibly pinched) $(d,k')$-dome gadget on player set~$P$ for some integers~$d$ and~$k'$ satisfying $n\geq k'(d+1)$ and $k'>2d+3$.
Assume further that no player in~$P$ except for its top player~$p^\star$ has a friend outside~$P$.
If $P \in \Gamma$ is a coalition in some coalition structure~$\Gamma$, and a coalition $C$ blocks~$\Gamma$, then $C$ contains no player from the base clique of the gadget. 
\end{proposition}

\begin{proof}
Let $k''=k'-d-1$ denote the size of the base clique~$K$ in our gadget over~$P$.
First observe that if some non-fringe player~$i$ within the base clique~$K$ is in~$C$, then all~$k''$ players in the base clique are in~$C$, as otherwise each friend of~$i$ (and $i$ itself) loses a friend in~$C$ when compared to~$P$, and thus the average number of friends among $i$'s friends decreases by at least~$1$.
The average number of enemies may only decrease by at most~$d+1$, so $n>d+1$ implies $\util_i^\varphi(\Gamma_{C \to \emptyset})<\util_i^\varphi(\Gamma)$, a contradiction to~$i \in C$.
In fact, $i \in C$ further implies that each player in the base clique~$K$ must have the same number of friends in~$C$ as in~$\Gamma$, as otherwise the average number of friends among $i$'s friends (possibly counting also $i$ itself) decreases by at least~$\frac{1}{k''}$, while the average number of enemies of these players can only decrease by at most~$d+1$, so $\frac{n}{k''}>d+1$ implies $\util_i^\varphi(\Gamma_{C \to \emptyset})<\util_i^\varphi(\Gamma)$. 
Thus we have $P \setminus \{p^\star\} \subseteq C$. 
The presence of some fringe player~$j$ in~$C$ further implies $p^\star \in C$ by the same arguments, since $p^\star \notin C$ would mean a decrease of at least~$\frac{1}{k''+1}$ in the utility of~$j$. 
Therefore, we get $P \subseteq C$; however, then the utility of every player in the base clique is at most its utility in~$\Gamma$, which contradicts our assumption that $C$ blocks~$\Gamma$. 

We thus have proven that $C$ cannot contain non-fringe players in the base clique of~$P$.
However, from this it follows that no fringe player~$i$ of~$P$ can be contained in~$C$, as its utility in~$\Gamma_{C \to \emptyset}$ would be less than its utility in~$\Gamma$: Observe that all fringe-player friends of~$i$ (if any) lose at least one friend when switching from~$\Gamma$ to~$\Gamma_{C \to \emptyset}$ (because there are $k'-2d-1>0$ non-fringe players in the base clique~$K$), while the mid-player friend of~$i$ has at most~$d+1$ friends, which is strictly fewer than the $k''-1 \geq d+2$ friends of the base clique players in~$\Gamma$. 
Thus, no player in~$K$ can be contained in~$C$.
\end{proof}

We are now ready to show $\conp$-completeness of verifying core stability in an \AHG\ with either \avgEQ\ or \avgAL\ preferences, starting with the former.

\begin{theorem}
\label{thm:avgEQ-corestable-coNPc}
Verifying core stability in an average-based EQ \AHG\ is $\conp$-complete.
\end{theorem}

\begin{proof}
  Membership of the verification problem in $\conp$ again is obvious.
  To show its $\conp$-hardness, we again reduce from \myproblem{Clique}.
  Let $(H,k)$ be our input, where $H=(V,E)$ is a graph and~$k$ an integer.
  We may assume, without loss of generality, that $k \geq 4$ and that $|V|+|E|\geq k^2$.
  Let $k'=k+3\binom{k}{2}+1$.
  The reduction will, to some degree, be similar to the one given in the proof of Theorem~\ref{thm:minEQAL-corestable-coNPc}, but we need to carefully modify the construction and proof of correctness to make it work for  \avgEQ\ preferences.

\smallskip
\noindent
{\bf Construction.}
We construct vertex, edge, and incidence gadgets as follows.
For each vertex $v \in V$, we add a $(k-1,k')$-dome gadget
over player set~$P_v$ whose top player is the \emph{vertex player}~$v'$ corresponding to~$v$.
For each edge $e \in E$, we add a $(2,k')$-dome gadget 
over player set~$P_e$ whose top player is the \emph{edge player}~$e'$ corresponding to~$e$.
Further, for each $v \in V$ and~$e \in \delta(v)$, we introduce a $(2,k')$-dome gadget
over player set~$P_{v,e}$ whose top player is the \emph{incidence player}~$b_{v,e}$.
Let $n$ denote the total number of players in the constructed gadgets, and 
we will also use the notation
$V' = \{v' \mid v \in V\}$, $E' = \{e' \mid e \in E\}$, and $B = \{b_{v,e} \mid v \in V, e \in \delta(v)\}$. 
Besides the friendships within gadgets, we let each incidence player~$b_{v,e} \in B$ be friends with the vertex player~$v'$ and the edge player~$e'$, so the friendship graph induced by $V' \cup E' \cup B$ can be obtained from~$H$ by subdividing each of its edges once.
This completes the definition of our average-based EQ \AHG.

Let $\Gamma$ be the coalition structure that contains
the player set of each vertex, edge, and incidence gadget as a coalition;
formally,
\[
\Gamma = \{P_v \mid v \in V\} \cup \{P_e \mid e \in E\}  \cup \{P_{v,e} \mid v \in V,e \in \delta(v)\}.
\] 

Due to Proposition~\ref{prop:dome}, we will not be interested in non-fringe base players, so we only need to compute the valuation of the remaining players in~$\Gamma$, which are as follows:
\[
\val_i(\Gamma)=\left\{
\begin{array}{ll}
(k-1)\cdot n - (k'-k) & \text{ if $i$ is a vertex player,} \\
2\cdot n - (k'-3) & \text{ if $i$ is an edge, an incidence, or a mid player,} \\
(k'-k) \cdot n - (k-1) & \text{ if $i$ is a fringe player in a vertex gadget,} \\
(k'-3) \cdot n - 2 & \text{ if $i$ is a fringe player in an edge or incidence gadget.} 
\end{array}
\right.
\]
It is now straightforward to compute the utilities of the players in~$\Gamma$ in the  average-based EQ model; we provide these for top and mid players below:
\[
\util_i^\avgEQ(\Gamma)=\left\{\!\!
\begin{array}{l@{\hspace{5pt}}l}
\frac{3k-3}{k}\cdot n - \left(k'-1-\frac{3k-3}{k}\right) & \text{if $i \in V'$,} \\[2pt]
2 \cdot n - (k'-3) & \text{if $i \in E' \cup B$,} \\[2pt]
\frac{k'+1}{3} \cdot n - (k'-1-\frac{k'+1}{3}) & \text{if $i$ is a mid player in $\Gamma(i)$.}
\end{array}
\right.
\]

\smallskip
\noindent
{\bf Proof of correctness.}
We claim that $\Gamma$ is \emph{not} core-stable if and only if $H$ contains a clique of size~$k$.

\proofonlyif
Let us first assume that $\Gamma$ is \emph{not} core-stable; we show that $H$ contains a clique of size~$k$. 
Let $C$ be a coalition blocking~$\Gamma$.
Observe that $C$ can contain no mid players from an edge or incidence gadget: By Proposition~\ref{prop:dome}, a mid player~$i$ can have only one friend in~$C$, namely the top player in~$\Gamma(i)$, while the top player in such a gadget has four friends in total; hence, the utility of~$i$ in~$C$ can be at most~$\frac{5}{2}n < \util_i^\avgEQ(\Gamma)$, a contradiction to~$i \in C$. 
Therefore, by Proposition~\ref{prop:dome}, $C$ can only contain players in~$V' \cup E' \cup B$, and possibly some mid players from vertex gadgets.

Next, we show an analogue of Claim~\ref{clm:min-degree}.

\begin{claim} 
\label{clm:dome-min-degree}
If $e' \in C$ for some edge~$e=xy \in E$, then $\{b_{x,e},b_{y,e},x',y' \}\subseteq C$ and  $|C|<k'$.
\end{claim}

\begin{claimproof} 
  The friends of~$e'$ in~$C$ can only be the two incidence players $b_{x,e}$ and~$b_{y,e}$, and they both can have at most two friends in~$C$ (besides~$e'$,
  one of the vertex players~$x'$ and~$y'$).
  Thus the utility of~$e'$ in $\Gamma_{C \to \emptyset}$ can exceed~$\util_{e'}^\avgEQ(\Gamma)$ only if $\{b_{x,e},b_{y,e}\} \subseteq C$ and, moreover, both $b_{x,e}$ and~$b_{y,e}$ have two friends in~$C$, which leads to $\{x',y'\} \subseteq C$.
Therefore, the utility of~$e'$ in $\Gamma_{C \to \emptyset}$ is exactly $\util_{e'}^\avgEQ(\Gamma_{C \to \emptyset})=2n-(|C|-3)$, which exceeds $\util_{e'}^\avgEQ(\Gamma)$ if and only if $|C|<k'$.
 \end{claimproof}

Let us create a subgraph~$H_C$ of~$H$ that contains a given vertex or edge if and only if the corresponding vertex or edge player is contained in~$C$. 
By Claim~\ref{clm:dome-min-degree}, $H_C$ is a well-defined subgraph of~$H$.

We proceed by showing that every vertex~$v$ in~$H_C$ has at least~$k-1$ incident edges.
Intuitively, the reason for this is that adding friends with only one friend in the coalition cannot raise the utility of a player.
For a formal proof, assume for the sake of contradiction that $v$ has at most~$k-2$ incident edges.
Then $v'$ has at most $k-2$ friends with two friends being in~$C$, and $v'$ may have an additional number~$\ell$ of friends, each of whom has only~$v'$ as a friend in~$C$.
Then we get
\begin{align*}
\util_{v'}^\avgEQ(\Gamma_{C \to \emptyset}) &\leq 
\frac{2(k-2)+(k-2+\ell)+\ell}{k+\ell-1} \cdot n 
= \left( 3-\frac{\ell+3}{k+\ell-1} \right) \cdot n \\
&  < \left(\frac{3k-3}{k}-\frac{1}{k^2}\right) \cdot n<
\frac{3k-3}{k}\cdot n - \left(k'-1-\frac{3k-3}{k}\right)= \util_{v'}^\avgEQ(\Gamma),
\end{align*}
where the first strict inequality holds because $k\geq 4$ and $\ell\geq 0$, as can be checked through simple
calculation, and the second strict inequality follows from our assumptions on the size of the graph~$H$ that guarantees $n \geq k^2\cdot k'$.
Therefore, we obtain a contradiction to~$v' \in C$, proving that each vertex in~$H_C$ has degree at least~$k-1$.

To show that $C$ is a clique, we use Observation~\ref{obs:clique}.
Recall that $|C|< k'$ due to Claim~\ref{clm:dome-min-degree}. Since each incidence player has only one edge-player friend, but each edge player in~$C$ has two incidence-player friends in~$C$, we get 
\begin{equation}
\label{eqn:bound-on-blockingcoal}
k+3\binom{k}{2}=k'-1 \geq |C| = |C \cap V|+|C \cap E|+|C \cap B| \geq |C \cap V|+3|C \cap E|.
\end{equation}
Applying Observation~\ref{obs:clique} with $\alpha=3$, it follows that $H_C$ is a clique of size~$k$.

\proofif
Let us now show that if $H$ contains a clique on a set~$K$ of $k$ vertices, then there exists a coalition~$C$ that blocks~$\Gamma$.
Similarly as in the proof of Theorem~\ref{thm:minEQAL-corestable-coNPc}, let $C$ contain the vertex and edge players corresponding to the vertices and edges of the clique~$K$, together with all friends of these edge players in~$B$, so $|C|=k+3\binom{k}{2}$.
In the new coalition structure~$\Gamma_{C \to \emptyset}$ (after the members of $C$ have deviated from their coalitions in~$\Gamma$), the utility of each player~$i$ in~$C$ is as follows:
\[
\util_i^\avgEQ(\Gamma_{C \to \emptyset})=\left\{\!\!
\begin{array}{ll}
\frac{3k-3}{k}\cdot n - \left(|C|-1-\frac{3k-3}{k}\right) & \text{if $i \in V'$,} \\[2pt]
2 \cdot n - (|C|-3) & \text{if $i \in E'$,} \\[2pt]
\frac{k+3}{3} \cdot n - (|C|-1-\frac{k+3}{3}) & \text{if $i \in B$.}
\end{array}
\right.
\]

Since $|C|=k+3\binom{k}{2}<k'$ and $k\geq 4$, it follows that $\util_i^\avgEQ(\Gamma_{C \to \emptyset}) < \util_i^\avgEQ(\Gamma)$ for each player~$i$ in~$C$. 
Therefore, $C$ blocks~$\Gamma$, and so $\Gamma$ is \emph{not} core-stable.
\end{proof}

Finally, we turn to showing that verifying core stability in an average-based AL \AHG\ is $\conp$-complete. 
The proof of Theorem~\ref{thm:avgAL-corestable-coNPc}
is similar to the proof Theorem~\ref{thm:avgEQ-corestable-coNPc} and relies on pinched dome-gadgets.

\begin{theorem}
\label{thm:avgAL-corestable-coNPc}
Verifying core stability in an average-based AL \AHG\ is $\conp$-complete.
\end{theorem}

\begin{proof}
  Again, membership of the verification problem in $\conp$ is obvious and, to show its $\conp$-hardness, we present a reduction from \myproblem{Clique}.
  The proof will be quite similar to that of
  Theorem~\ref{thm:avgEQ-corestable-coNPc}: We use the same ideas and an analogous construction, however, the details need to be carefully adjusted.
  Let $(H,k)$ be our input instance of \myproblem{Clique}, where $H=(V,E)$ is a graph and~$k$ an integer.
  Set $k'=k+3\binom{k}{2}+1$.
  Without loss of generality, we assume that $|V|+|E| > k \geq 3$ and that $\frac{k+1}{3}$ is an integer.

\smallskip
\noindent
{\bf Construction.}
We construct vertex, edge, and incidence gadgets as follows.
For each vertex $v \in V$, we introduce a pinched $(2,k')$-dome gadget over player set~$P_v$ whose top player is the \emph{vertex player} $v'$ corresponding to~$v$.
For each edge $e \in E$, we introduce a pinched $(2,k')$-dome gadget over player set~$P_e$ whose top player is the \emph{edge player} $e'$ corresponding to~$e$.
Additionally, for each vertex $v \in V$ and each edge $e \in \delta(v)$, we introduce a pinched $\left(\frac{k+1}{3},k'\right)$-dome gadget
over player set~$P_{v,e}$ whose top player is the \emph{incidence player}~$b_{v,e}$.
We let $n$ denote the total number of players in the constructed gadgets, and we will use the notation $V' = \{v' \mid v \in V\}$, $E' = \{e' \mid e \in E\}$, and $B = \{b_{v,e} \mid v \in V, e \in \delta(v)\}$. 
Besides the friendships within gadgets, we let each incidence player~$b_{x,e} \in B$ for some edge~$e=xy \in E$ be friends with the vertex player~$x'$, the edge player~$e'$, and---in addition---the incidence player~$b_{y,e}$. 
This completes the definition of our \AHG.

Let $\Gamma$ be the coalition structure that contains
the player set of each vertex, edge, and incidence gadget as a coalition;
formally,
\[
\Gamma = \{P_v \mid v \in V\} \cup \{P_e \mid e \in E\}  \cup \{P_{v,e} \mid v \in V,e \in \delta(v)\}.
\] 

Due to Proposition~\ref{prop:dome}, we will only be interested in the utilities of the top and mid players within each gadget in~$\Gamma$.
Recall that the mid player in a pinched $(d,k')$-gadget has $d+1$ friends in the gadget.
Simple calculation and $k \geq 3$ gives us the following:
\begin{align*}
\util_i^\avgAL(\Gamma) & =\left\{ 
\begin{array}{ll}
3 \cdot n - (k'-4) & \text{ if } i \in V' \cup E', \\[2pt]
\frac{k+4}{3} \cdot n - \left(k'-1-\frac{k+4}{3}\right) \phantom{kki} & \text{ if }i \in B'; 
\end{array} \right.
\\
\util_i^\avgAL(\Gamma) & \geq 
\begin{array}{ll}
\phantom{m}
\frac{2k'-3}{3} \cdot n - \left(k'-1-\frac{2k'-3}{3}\right) & \text{ if }i \text{ is a mid player in~$\Gamma(i)$}.
\end{array}
\end{align*}

\smallskip
\noindent
{\bf Proof of correctness.}
We claim that $\Gamma$ is \emph{not} core-stable if and only if $H$ contains a clique of size~$k$.

\proofonlyif
Assume first that some coalition $C$ blocks~$\Gamma$.
Observe that a blocking coalition $C$ can contain no mid players from an edge or incidence gadget: By Proposition~\ref{prop:dome}, a mid player~$i$ can have only one friend in~$C$, namely the top player in~$\Gamma(i)$, while the top player in such a gadget has at most four friends in total; hence, the utility of~$i$ in~$C$ can be at most~$4n < \frac{2k'-4}{3} \cdot n < \util_i^\avgAL(\Gamma)$, a contradiction to~$i \in C$. 
Therefore, by Proposition~\ref{prop:dome}, $C$ can only contain players in~$V' \cup E' \cup B$, and possibly some mid players from vertex gadgets.

We define the graph~$H_C$ as in the proof of Theorem~\ref{thm:avgEQ-corestable-coNPc}, i.e., $H_C$ contains a given vertex or edge of~$H$ if and only if the corresponding
vertex or edge player is in~$C$.

Next, we show that Claim~\ref{clm:dome-min-degree} remains true for the modified construction.

\begin{claim} 
\label{clm:pinched-dome}
If $e' \in C$ for some edge~$e=xy \in E$, then $\{b_{x,e},b_{y,e},x',y' \}\subseteq C$ and  $|C|<k'$.
\end{claim}

\begin{claimproof} 
  The friends of~$e'$ in~$C$ can only be the two incidence players $b_{x,e}$ and~$b_{y,e}$, and both of these players can have at most three friends in~$C$.
  If one of $b_{x,e}$ and~$b_{y,e}$ is not in~$C$, then the utility of~$e'$ in~$\Gamma_{C \to \emptyset}$ could be at most~$2n < \util_{e'}^\avgAL(\Gamma)$.
  Hence, we get $\{b_{x,e},b_{y,e}\} \subseteq C$; moreover, both $b_{x,e}$ and~$b_{y,e}$ must have three friends in~$C$, which leads to $\{x',y' \} \subseteq C$.
Therefore, the utility of~$e'$ in $\Gamma_{C \to \emptyset}$ is exactly $\util_{e'}^\avgAL(\Gamma_{C \to \emptyset}) = 3n-(|C|-4)$, which exceeds $\util_{e'}^\avgAL(\Gamma)$ if and only if $|C|<k'$.
\end{claimproof}

By Claim~\ref{clm:pinched-dome}, we know that $H_C$ is a well-defined subgraph of~$H$. We next prove that every vertex  in~$H_C$ has degree at least~$k-1$. 

Assume $x' \in C$. First notice that each friend of~$x'$ has at most three friends in~$C$. Thus all friends of~$x'$ in~$C$ need to have \emph{exactly} three friends in~$C$, as otherwise the utility of~$x'$ in~$\Gamma_{C \to \emptyset}$ would be 
lower than~$\frac{2+3(\ell-1)}{\ell}\cdot n=3n-\frac{n}{\ell} <3n-k' < 
\util_{x'}^\avgAL(\Gamma)$
where $\ell$ denotes the number of friends that~$x'$ has in~$C$; 
notice that the inequality $n>k'\cdot \ell$ we use here follows from the fact that $x$ must have degree at least~$\ell-1$ in~$H$. 
In particular, no mid player in the vertex gadget containing~$x'$ can be in~$C$. Furthermore, whenever $b_{x,e} \in C$ for some edge~$e$ incident to~$x$ in~$H$, then $e' \in C$ follows. This means that the degree of~$x$ in~$H_C$ is exactly the number of friends that $x'$ has in~$C$. 

Consider now
the utility of some incidence player
$b_{x,e}$ in~$C$: If $x'$ has at most~$k-2$ friends in~$C$, then
\[
\util_{b_{x,e}}^\avgAL(\Gamma_{C \to \emptyset}) \leq
\frac{2+3+k-2}{3} \cdot n <
\frac{k+4}{3}\cdot n-k'<
\util_{b_{x,e}}^\avgAL(\Gamma)
\]
where the first strict inequality follows from
$n>3k'$ which in turn holds due to our assumptions on the size of the graph~$H$.
This proves that $x'$ needs to have at least~$k-1$ friends in~$C$, 
that is, every vertex in~$H_C$ has degree at least~$k-1$. 
Finally, observe that the inequality~(\ref{eqn:bound-on-blockingcoal}) holds; using also Claim~\ref{clm:pinched-dome}, we can apply Observation~\ref{obs:clique} with $\alpha=3$. It follows that $H_C$ is a clique of size~$k$.

\proofif
Given a
clique of size~$k$ in~$H$, proving that the corresponding vertex, edge, and incidence players form a
coalition blocking~$\Gamma$ is a straightforward adaptation of the arguments presented in the proof of Theorem~\ref{thm:avgEQ-corestable-coNPc}.
\end{proof}

\section{Conclusions and Open Questions}

Having solved the last four open problems related to the computational complexity of verifying core stability in altruistic hedonic games, the picture for this property is now complete: For all three degrees of altruism, in both the average-based and the minimum-based case, and for both altruistic hedonic games and the more general altruistic CFGs, it is $\conp$-complete to verify whether a given coalition structure is core-stable.
Of course, many related problems remain open and can be tackled in future research.

For example, we only have an upper bound of containment in~$\conp$ for verifying \emph{strict} core stability, and it remains to show a matching lower bound of $\conp$-hardness.
A coalition structure $\Gamma$ is \emph{strictly core-stable} if it is not weakly blocked by any coalition, i.e., for each coalition $C \subseteq N$, we either have $\Gamma(i) \succ_i \Gamma_{C\rightarrow\emptyset}(i)$ for some player $i \in C$, or we have $\Gamma(i) \sim_i \Gamma_{C\rightarrow\emptyset}(i)$ for all players $i \in C$.
Furthermore, the existence problems for core stability and strict core stability have not been classified in terms of their complexity yet, in any of the models of AHGs or ACFGs we have considered here.

In addition, the existence problems for other properties of AHGs or ACFGs remain open as well, such as the existence of a (strictly) popular coalition structure, even though the corresponding verification problems have recently been settled for AHGs~\cite{ker-rot:j:complexity-of-verifying-popularity-and-strict-popularity-in-altruistic-hedonic-games}.
A coalition structure $\Gamma$ is \emph{popular} if it is preferred to any other coalition structure $\Delta$ by at least as many players as there are players preferring $\Delta$ to~$\Gamma$.

Finally, in addition to classical complexity, it would be very interesting to study these problems in terms of their fixed-parameter tractability and parameterized complexity (see, e.g., the work of Chen \emph{et al.}~\cite{che-csa-roy-sim:c:hedonic-games-with-friends-enemies-and-neutrals}) or in terms of their approximability (see, e.g., the work of Munagala \emph{et al.}~\cite{mun-she-wan:c:auditing-for-core-stability-in-participatory-budgeting} who study core stability in the context of participatory budgeting).

\subsection*{Acknowledgments}

This work was supported in part by Deutsche Forschungsgemeinschaft under DFG research grant \linebreak[4]RO\nobreakdash-1202/21\nobreakdash-2 (project 438204498) and by the Hungarian Academy of Sciences under its Momentum Programme (LP2021\nobreakdash-2) and its J{\'a}nos Bolyai Research Scholarship.

\bibliographystyle{eptcs}
\bibliography{joergbib}



\end{document}